\documentclass[aps,pra,reprint,10pt,a4paper,nofootinbib,showpacs]{revtex4-1}
\usepackage[utf8]{inputenc}
\usepackage{amsfonts}
\usepackage{amssymb}
\usepackage{amsmath}
\usepackage{amsthm}
\usepackage{graphics}
\usepackage{graphicx}
\usepackage{braket}
\usepackage{ulem}
\usepackage[colorlinks=true,linkcolor=blue,urlcolor=blue,citecolor=blue]{hyperref}
\newtheorem{theorem}{Theorem}

\newtheorem{definition}{Definition}
\newtheorem{proposition}{Proposition}

\newtheorem{observation}{Observation}

\begin{document}
\title{Local indistinguishability and incompleteness of entangled orthogonal bases:\\Method to generate two-element locally indistinguishable ensembles}

\author{Saronath Halder}
\affiliation{Harish-Chandra Research Institute, HBNI, Chhatnag Road, Jhunsi, Allahabad 211 019, India}

\author{Ujjwal Sen}
\affiliation{Harish-Chandra Research Institute, HBNI, Chhatnag Road, Jhunsi, Allahabad 211 019, India}

\begin{abstract}
We relate the phenomenon of local indistinguishability of orthogonal states with the properties of unextendibility and uncompletability of entangled bases for bipartite and multipartite quantum systems. We prove that all two-qubit unextendible entangled bases are of size three and they cannot be perfectly distinguished by separable measurements. We identify a method of constructing two-element orthogonal ensembles, based on the concept of unextendible entangled bases, that can potentially lead to information sharing applications. Two-element ensembles form the fundamental unit of ensembles, and yet does not offer locally indistinguishable ensembles for pure state elements. Going over to mixed states does open this possibility, but can be difficult to identify. The method provided using unextendible entangled bases can be used for their systematic generation. In multipartite systems, we find a class of unextendible entangled bases for which the unextendibility property remains conserved across all bipartitions. We also identify nonlocal operations, local implementation of which require entangled resource states from a higher-dimensional quantum system.  
\end{abstract}

\maketitle

\section{Introduction}\label{sec1}
A composite quantum system, distributed among several spatially separated parties, can exhibit ``nonlocal'' features \cite{Nielsen00}. The tagging of ``nonlocality'' is exercised on a rather broad school of phenomena. There exist certain nonlocal properties for which entanglement is necessary \cite{Brunner14}. Prominent examples of this group are violation of Bell inequalities \cite{Bell01} and quantum teleportation \cite{Bennett93, Collins01}. Interestingly, Bennett {\it et al.} demonstrated a type of nonlocality which can be exhibited by a set of orthogonal product states \cite{Bennett99-1}. Such nonlocality-exhibiting sets of orthogonal product states may or may not be able to form a complete orthogonal product basis. An unextendible product basis \cite{Bennett99, Divincenzo03} is among the examples of such sets, the orthogonal product states within which cannot be appended with any further orthogonal product state, and again, they can exhibit nonlocality. Thus, the states within an unextendible product basis span a subspace of a tensor-product Hilbert space such that the complementary subspace has no product state. 

In this work, the labelling of ``nonlocality'' is applied on the following phenomenon: Given a composite quantum system, we consider that the system is distributed among several spatially separated parties. We also assume that the system is prepared in a state, taken from a known set of orthogonal states. The task is to identify the state of the system under local quantum operations and classical communication. If it is possible to identify the state of the system correctly, then the states of the known set are ``locally distinguishable''. Otherwise, we say that the states of the chosen set are ``locally indistinguishable'', and since they are actually mutually orthogonal, their local indistinguishability is labelled as exhibiting a type of  ``nonlocality''. In case  the state of a given ensemble cannot be perfectly identified, it is natural to attempt a conclusive identification with some nonzero probability \cite{Chefles98, Chefles04, Ji05, Duan07, Walgate08, Bandyopadhyay09, Cohen14}. If even such a conclusive local identification is not possible, the corresponding the ensemble is considered to possess a ``higher'' or ``more'' nonlocality than if it were only deterministic-locally indistinguishable. Let us also  mention here that because of the complex structure of the set of physical maps implementable by local quantum operations and classical communication \cite{Chitambar14}, one sometimes uses separable measurements \cite{Duan09} to learn about necessary conditions of local distinguishability. In Refs.~\cite{Bennett99, Divincenzo03, Rinaldis04}, proofs related to local indistinguishability of states within  unextendible product bases were discussed. The study of the unextendibility property is important to understand the present nonlocal feature of distributed quantum systems. This nonlocal feature is relevant in practical applications, and e.g., it is a key ingredient in  protocols of quantum cryptography such as secret sharing \cite{Markham08}, data hiding \cite{Terhal01, Eggeling02}, etc.

It is probably useful to comment here on the use of the term ``basis'' in this paper. While a typical monograph in linear algebra defines a basis as a collection of vectors that is complete and linearly independent \cite{Simmons04, Gupta16}, we will here use the term ``basis'' even for incomplete sets, in accordance with the practice in the literature. Other such ``off-track'' use of the word, include the ``overcomplete basis'' of coherent states \cite{Mandel95} and ``basis of linearly dependent states'' \cite{Srivastava20}. 

After the discovery of unextendible product bases, the concept of unextendibility was generalized to the case of entangled states also. Unextendible bases using orthogonal maximally entangled states were introduced in Ref.~\cite{Bravyi11} for certain square dimensional systems ($\mathbb{C}^d\otimes\mathbb{C}^d$, $d=3,4$). A maximally entangled state of a bipartite quantum system is a pure state of that system that has the maximal number of Schmidt coefficients for the system, and these coefficients are all equal. For an unextendible maximally entangled basis of a tensor product of two Hilbert spaces, the orthogonal maximally entangled states within that basis span a proper subspace of the considered tensor-product Hilbert space, while the complementary subspace contains no maximally entangled state. Such bases are important to demonstrate the violation of the quantum Birkhoff conjecture \cite{Yu17-1}. Thereafter,  unextendible maximally entangled bases were constructed in Ref.~\cite{Chen13} on nonsquare dimensions ($\mathbb{C}^{d_1}\otimes\mathbb{C}^{d_2}$, $d_1d_2>4$, $d_2/2<d_1<d_2$). In the same paper, the concept of mutually unbiased unextendible maximally entangled bases was introduced. There are many other articles \cite{Li14, Wang14, Nan15, Nizamidin15, Guo16, Zhang16-3, Wang17-1, Zhang18-1, Zhang18-2, Liu18, Song18, Zhao20} which include discussions regarding bipartite unextendible maximally entangled bases. 

Bipartite unextendible entangled bases with fixed Schmidt rank-$k$ were constructed in Refs.~\cite{Guo14, Han18, Shi19, Yong19, Wang19}. For such an entangled basis, the orthogonal states of Schmidt rank-$k$ within that basis span a proper subspace of the considered Hilbert space, while the complementary subspace contains no state whose Schmidt rank is $\geq k$. A different type of bipartite unextendibility for nonmaximally entangled states and their application in a communication protocol were discussed in Refs.~\cite{Chakrabarty12,Chen13-1}. 

For the bipartite case, the labelling of an ``unextendible entangled basis'' is used for a set of mutually orthogonal states on a tensor product of two Hilbert spaces, where the states can be maximally or nonmaximally entangled ones, and they span a proper subspace of the considered tensor-product Hilbert space such that the complementary subspace contains only product states. Moreover, there is no restriction on the Schmidt ranks of the entangled states in an unextendible entangled basis.

Unextendibility for entangled states in {\it multipartite} quantum systems does not have a significant presence in the literature. In Refs.~\cite{Guo15-1, Zhang17-4}, a few multipartite cases were discussed. In Ref.~\cite{Guo15-1}, the authors proved that using the standard Greenberger-Horne-Zeilinger states \cite{Greenberger07, Mermin90}, it is not possible to construct a three-qubit unextendible entangled basis. They further conjectured that this will remain true even when the number of qubits is greater than three. However, they also constructed certain unextendible entangled bases for multipartite systems. In Ref.~\cite{Zhang17-4}, using a different proof technique, it was again shown that there is no unextendible entangled basis in a three-qubit Hilbert space when the states are standard Greenberger-Horne-Zeilinger states. Then the authors provided examples of unextendible entangled bases for higher dimensional tripartite quantum systems.  

The main focus in the existing literature on unextendible entangled bases has been  on different types of constructions. General properties have remained largely unexplored. In fact, in many research works, unextendibility of orthogonal local unitary operators is used to construct unextendible maximally entangled bases. But this technique has its limitations, particularly when one thinks about constructing unextendible entangled bases in multipartite systems. More precisely, in a multipartite system, entanglement has complex structures as there are different types of states such as fully separable states, biseparable states, and the genuinely multipartite entangled states \cite{Horodecki09-1, Guhne09, Das17}. Again, among the genuinely multipartite entangled states, there are inequivalent classes under stochastic local quantum operations and classical communication \cite{Dur00, Verstraete02, Miyake03}. Therefore, if one thinks about constructing an unextendible entangled basis which contains genuine multipartite entangled states from different inequivalent classes, then the technique of using unextendibility of orthogonal local unitary operators does not work. 

In this work, we consider different types of unextendible entangled bases for both bipartite and multipartite systems, and discuss properties of those bases. 

\begin{itemize}
\item[(a)] Particularly, for a two-qubit system, we provide certain constructions of unextendible entangled bases. Thereafter, we show that there is only one type of unextendible entangled base with respect to their cardinality for a two-qubit system, viz., unextendible entangled bases of size three. Furthermore, these bases cannot be perfectly distinguished by separable measurements and therefore, by local quantum operations and classical communication. In this context, it is good to mention that in this work, within all discrimination processes, the given states are equally probable. We also discuss about unextendible entangled bases in higher dimensional bipartite systems. 

\item[(b)] Apart from unextendible entangled bases, we analyze other types of incomplete entangled bases, viz., uncompletable and strongly uncompletable bases considering both maximally and nonmaximally entangled states. Moreover, based on the concept of bipartite unextendible entangled bases, we construct and analyze interesting ensembles, including two-element ensembles of locally indistinguishable orthogonal (mixed) states, which can be potential candidates for information processing (we encode the classical information against the possible states of a system and by identifying the state of the system, we decode that information). 

\item[(c)] For multipartite systems, we first consider a three-qubit system and construct two different unextendible entangled bases. The first basis contains only W states \cite{Zeilinger92, Dur00, Sen(De)03} while the second one contains W states as well as Greenberger-Horne-Zeilinger states. 

\item[(d)] We report that both the bases have an interesting property, viz., the unextendibility property remains conserved across every bipartition. We also mention that this is impossible for any unextendible product basis in a multi-qubit system. 

\item[(e)] Again, the second type of basis leads to a nonlocal operation, to implement which locally, one requires entangled resource states from a higher-dimensional Hilbert space. 

\item[(f)] An important property associated with the second type of basis is that a subset of five states of the basis can show local indistinguishability across every bipartition. 

\item[(g)] We also present an algorithm to construct unextendible entangled bases for any number of qubits, which can lead to nonlocal operations, to implement which locally, entangled resources from higher-dimensional Hilbert spaces are required. 

\item[(h)] We also prove that for three qubits, there are only two types of unextendible entangled bases (which are unextendible across every bipartition) with respect to their cardinalities, viz., unextendible entangled bases of sizes six and seven. 
\end{itemize}

The results regarding bipartite systems are given in Sec.~\ref{sec2}, while those on multiparty systems are given in Sec~\ref{sec3}.  Finally, in Sec.~\ref{sec4}, a conclusion is drawn. Some proofs are consigned to an Appendix.

\section{Bipartite systems}\label{sec2}
It is known that for a two-qubit system, there is no unextendible maximally entangled basis (UMEB) \cite{Bravyi11}. But a two-qubit unextendible entangled basis (UEB) can be constructed. In Ref.~\cite{Guo15-1}, it was shown that starting from a two-qubit UEB, it is possible to construct a three-qubit UEB. 

We first present here two different UEBs for a two-qubit system. Then, we talk about several important properties of those bases. The first one consists of the states 
\begin{equation}\label{eq1}
\begin{array}{l}
\frac{1}{\sqrt{3}}(\ket{00}+\ket{01}+\ket{10}),~
\frac{1}{\sqrt{2}}(\ket{01}-\ket{10}),\\[1.5 ex]
\frac{1}{\sqrt{3}}(\sqrt{2}\ket{00}-\frac{1}{\sqrt{2}}\ket{01}-\frac{1}{\sqrt{2}}\ket{10}).
\end{array}
\end{equation}
In this paper, we use the notation $\ket{v_1v_2\dots v_m}\equiv\ket{v_1}\otimes\ket{v_2}\otimes\dots\otimes\ket{v_m}$ for an $m$-partite quantum state. There is only one two-qubit state which is orthogonal to the above entangled states and that state is $\ket{11}$, which is a product state. This implies that the above entangled states form a UEB. An important feature of the above UEB is that the entangled states are not equally entangled, and in the computational basis,  the coefficients are all real. We now present a UEB which contains equally entangled states. It consists of the states 
\begin{equation}\label{eq2}
\begin{array}{l}
\frac{1}{\sqrt{3}}(\ket{00}+\ket{01}+\ket{10}),\\[1.5 ex]
\frac{1}{\sqrt{3}}(\ket{00}+\omega\ket{01}+\omega^2\ket{10}),\\[1.5 ex]
\frac{1}{\sqrt{3}}(\ket{00}+\omega^2\ket{01}+\omega\ket{10}),
\end{array}
\end{equation}
where $\omega$ is a nonreal cube root of unity. Notice that both the UEBs span the same subspace. But in case of the second basis, the coefficients in the computational basis are complex quantities. 

We next talk about uncompletability of sets of entangled states. The definition of uncompletability for product states was given in Ref.~\cite{Divincenzo03}. Following the same definition, we provide the definition of an uncompletable entangled basis (UCEB).
\begin{definition}\label{def1}
Given a set of orthogonal pure entangled states, we assume that the states span a proper subspace of a tensor-product Hilbert space. If it is possible to find a nonzero number of entangled states in the complementary space, which however are not sufficient to form a complete orthogonal entangled basis of the entire tensor-product  Hilbert space, then the given set is said to be an uncompletable entangled basis.
\end{definition}
Note that the two UEBs presented above has the same first element. Suppose now that we remove this first state from any of the sets.  Then any pure state from the complementary space can be written as a linear combination of the two states, $(1/\sqrt{3})(\ket{00}+\ket{01}+\ket{10})$ and $\ket{11}$. Here, it is possible to construct two orthogonal entangled states, adding which to the remaining two states in any of the above sets, a complete orthogonal entangled basis can be constructed. So, in this case, these two sets of two states do not constitute UCEBs. We will return later to the concept of uncompletability of entangled states again. However, for the two-qubit UEBs, we now present the following proposition.
\begin{proposition}\label{prop1}
Any two-qubit unextendible entangled basis consists of three entangled states, and they  cannot be perfectly distinguished by separable measurements. 
\end{proposition}

\begin{proof}
In general, a two-qubit orthogonal product basis can be written as $\{\ket{a0}, \ket{a1}, \ket{b0^\prime}, \ket{b1^\prime}\}$, where $\{\ket{a}, \ket{b}\}$, $\{\ket{0}, \ket{1}\}$, and $\{\ket{0^\prime}, \ket{1^\prime}\}$ are different orthogonal bases for a qubit system \cite{Walgate02}. Here, the states $\ket{0^\prime}$, $\ket{1^\prime}$ can be thought of as linear combinations of the orthogonal states $\ket{0}$ and $\ket{1}$, in such a way that $\ket{0^\prime}$, $\ket{1^\prime}$ are orthogonal to each other. Now, we can choose any three product states from the general two-qubit product basis, and taking suitable linear combinations, it might be possible to produce three orthogonal pure entangled states, which form a UEB. Without loss of generality, we can consider the chosen set of product states as $\ket{a0}$, $\ket{a1}$, and $\ket{b0^\prime}$. Now, consider the following three states:
\begin{equation}\label{eq3}
\begin{array}{l}
\frac{1}{\sqrt{3}}(\ket{a0} + \ket{a1} + \ket{b0^\prime}),\\[2 ex]
\frac{1}{\sqrt{3}}(\ket{a0} + \omega\ket{a1} + \omega^2\ket{b0^\prime}),\\[2 ex]
\frac{1}{\sqrt{3}}(\ket{a0} + \omega^2\ket{a1} + \omega\ket{b0^\prime}).
\end{array}
\end{equation}
Clearly, there is only one state left, which is orthogonal to the above three states, and the state is $\ket{b1^\prime}$, a product state. Now, if \(|0^\prime\rangle = |0\rangle\) and \(|1^\prime\rangle = |1\rangle\), the three states in (\ref{eq3}) are entangled, so, they form an unextendible entangled basis. 

Notice that two or less orthogonal pure entangled states in a two-qubit system cannot form a UEB. This is straightforward from the general structure of two-qubit orthogonal product bases. However, we provide here a brief proof. We first consider the case of two pure entangled states and we assume that those states form a UEB. So, there will only be product states in the complementary space. Let us consider two such pure product states, which are orthogonal, in the complementary subspace. By the assumption, any linear combination of the two product states cannot be entangled. So, the product states can be of the forms $\ket{l_1}\ket{l_2}$ and $\ket{l_1}\ket{l_2^\perp}$. In the span of the assumed UEB also, it is possible to think of two product states and both of them must be orthogonal to the states $\ket{l_1}\ket{l_2}$ and $\ket{l_1}\ket{l_2^\perp}$. From the general structure of the two-qubit product basis, it is clear that the product states in the span of the entangled states must have the forms $\ket{l_1^\perp}\ket{l_2^\prime}$ and $\ket{l_1^\perp}\ket{l_2^{\prime\perp}}$. But any linear combinations of $\ket{l_1^\perp}\ket{l_2^\prime}$ and $\ket{l_1^\perp}\ket{l_2^{\prime\perp}}$ cannot produce entangled states. Thus, it contradicts with the assumption that the two entangled states forms a UEB. Following similar arguments, it is also possible to prove that a single entangled state cannot form a UEB.

So, it is quite clear that for a two-qubit system, only a single cardinality is possible for UEBs, and that is three. We now use a criterion from Ref.~\cite{Duan09} which asserts that any three pure orthogonal two-qubit states $\ket{\Phi_1}$, $\ket{\Phi_2}$, $\ket{\Phi_3}$ (irrespective of whether they are entangled or product) cannot be perfectly distinguished by separable measurements if $\sum_{i=1}^3\mathcal{C}(\ket{\Phi_i})\neq\mathcal{C}(\ket{\Phi_4})$, where $\mathcal{C}(\cdot)$ is the concurrence \cite{Wootters98} of its argument and $\ket{\Phi_4}$ is the unique state orthogonal to the states $\ket{\Phi_i}$, $\forall i = 1,2,3$. When $\ket{\Phi_1}$, $\ket{\Phi_2}$, $\ket{\Phi_3}$ form a two-qubit UEB, $\mathcal{C}(\ket{\Phi_4})$ must be zero and $\sum_{i=1}^3\mathcal{C}(\ket{\Phi_i})$ must be nonzero. Thus, the states of a UEB cannot be perfectly distinguished by separable measurements. These complete the proof of the proposition.
\end{proof}

From the above considerations, it is evident that UEBs provide a {\it sufficient} criterion for  indistinguishability of three two-qubit pure entangled states under separable measurements (SEP), viz., if three entangled states form a two-qubit UEB, then the states cannot be perfectly distinguished by SEP. Furthermore, if a set of orthogonal quantum states cannot be perfectly distinguished by SEP, then it cannot be perfectly distinguished by local quantum operations and classical communication (LOCC) \cite{Bennett99-1}.

It is good to stress here that in this work, we only consider orthogonal quantum states, and until now we have considered  perfect distinguishability of such states under LOCC or SEP. However, we are also going to consider  probabilistic distinguishability later but in the conclusive way and under LOCC \cite{Chefles04}. 

We next discuss about UMEBs. In Ref.~\cite{Chen13}, a UMEB is constructed in the minimum nonsquare dimension. The construction shows that it is basically a $2\otimes2$  maximally entangled basis (MEB) which plays the role of a UMEB in $2\otimes3$. [We will henceforth use the notation $d_1\otimes d_2\otimes\dots\otimes d_m$ instead of $\mathbb{C}^{d_1}\otimes\mathbb{C}^{d_2}\otimes\dots\otimes\mathbb{C}^{d_m}$.] But a $2\otimes2$ MEB may not play the role of a UMEB in $2\otimes d$ when $d\geq4$. This can be understood in the following way. Suppose, there is a $2\otimes2$ MEB, given by four states,
\begin{equation}\label{eq4}
\frac{1}{\sqrt{2}}(\ket{00^\prime}\pm\ket{11^\prime}),~\frac{1}{\sqrt{2}}(\ket{01^\prime}\pm\ket{10^\prime}),
\end{equation}
where $\ket{0}$, $\ket{1}$ forms an orthogonal basis for a two-level quantum system on Alice's side and $\ket{0^\prime}$, $\ket{1^\prime}$ forms an orthogonal basis for a two-level quantum system on Bob's side, with Alice and Bob being the observers in possession of the two systems involved. If the extended Hilbert space is $2\otimes4$, then we can consider the orthogonal product states $\ket{0x}$, $\ket{0x^\prime}$, $\ket{1x}$, $\ket{1x^\prime}$, where $\ket{x}$, $\ket{x^\prime}$ along with $\ket{0^\prime}$ and $\ket{1^\prime}$ form a complete orthonormal basis of the four-dimensional side. So, now it is possible to construct the four mutually orthogonal maximally entangled states,
\begin{equation}\label{eq5}
\frac{1}{\sqrt{2}}(\ket{0x}\pm\ket{1x^\prime}),~\frac{1}{\sqrt{2}}(\ket{0x^\prime}\pm\ket{1x}),
\end{equation}
which, along with the states in (\ref{eq4}) form a complete MEB. In general, we have an observation related to the above, given as the following.
\begin{observation}\label{obs1}
Any complete maximally entangled basis in $d\otimes d$ is an unextendible maximally entangled basis in $d\otimes (d+n)$, where $n$ is an integer in \([1,d)\).
\end{observation}
\begin{proof}
For any value of $n$, one can consider the product states $\ket{i}\ket{j}$, where $i=0,1,\dots,(d-1)$ and $j=d,(d+1),\dots,(d+n-1)$. These product states are orthogonal to the states of the given MEB in $d\otimes d$. As long as $n<d$, it is not possible to construct an entangled state of Schmidt rank-$d$ using the product states $\ket{i}\ket{j}$. Therefore, the MEB in $d\otimes d$, behaves like a UMEB in $d\otimes (d+n)$.
\end{proof}

In Definition \ref{def1}, we have described uncompletability for entangled states. We now want to provide a definition of ``strong uncompletability'' for sets of entangled states. Like the notion of uncompletability, strong uncompletability was also introduced for product states in Ref.~\cite{Divincenzo03}. 
\begin{definition}\label{def2}
Consider an uncompletable entangled basis. If that uncompletable entangled basis cannot be completed even in any locally extended Hilbert space, then the given states form a strongly uncompletable entangled basis. 
\end{definition}
Note that a local extension of the Hilbert space can happen on any party's side, or on both. Note also that in both the definitions (Definitions \ref{def1} and \ref{def2}), if we replace entangled states by maximally entangled ones, then we get the notions of uncompletability and strong uncompletability for sets of maximally entangled states. (It should be remembered that while considering local extensions in case of maximally entangled bases, the local extensions must be carried out only on one side, to preserve the maximal entanglement property of the constituent states.) From Observation \ref{obs1}, it is quite clear that the UMEBs which are mentioned in that Observation, can be extended to a complete MEB in a sufficiently locally extended Hilbert space. In fact, the technique in Observation \ref{obs1} cannot be used to construct strongly uncompletable maximally entangled bases (SUCMEB), as an uncompletable maximally entangled basis can always be completed to a MEB in some locally extended Hilbert space. However, there is another interesting observation which can be extracted from the $2\otimes3$ UMEB, given in (\ref{eq4}). This observation is presented as the following:
\begin{proposition}\label{prop2}
If a state, let us say, $(1/\sqrt{2})(\ket{00^\prime}+\ket{11^\prime})$, is removed from the $2\otimes3$ UMEB, given in (\ref{eq4}), then it is not possible to get sufficient (in this case, three) pairwise orthogonal pure maximally entangled states from the rest of the Hilbert space to complete the basis. 
\end{proposition}
\begin{proof}
After the removal of the state, \((1/\sqrt{2})(\ket{00^\prime}+\ket{11^\prime})\), the orthogonal complement of the space spanned by the other three states in (\ref{eq4}), is spanned by \((1/\sqrt{2})(\ket{00^\prime}+\ket{11^\prime})\), and the two product states, \(|02^\prime\rangle\) and \(|12^\prime\rangle\). (Note that $\ket{0}$ and $\ket{1}$ are forming an orthogonal basis for the qubit Hilbert space, and on the other hand, $\ket{0^\prime}$, $\ket{1^\prime}$, and $\ket{2^\prime}$ are forming an orthogonal basis for the qutrit Hilbert space.) Let us now consider an arbitrary linear superposition of these three states in the orthogonal complement, viz. \((e/\sqrt{2})(\ket{00^\prime}+\ket{11^\prime}) + f|02^\prime\rangle+ g|12^\prime\rangle\), with \(|e|^2 + |f|^2 + |g|^2 =1\). If this state has to be maximally entangled, its local density on the qubit side must be maximally mixed. Forcing that constraint results in vanishing \(f\) and \(g\). Therefore, the orthogonal complement can support only a single maximally entangled state. This completes the proof.
\end{proof}

However, it is possible to construct three pairwise orthogonal nonmaximally entangled states to complete the basis. So, the set of the remaining three states after $(1/\sqrt{2})(\ket{00^\prime}+\ket{11^\prime})$ is removed, is an uncompletable maximally entangled basis (UCMEB), but not a UCEB. Here, the UCMEB is locally indistinguishable, as in $2\otimes2$, three entangled states are always locally indistinguishable \cite{Walgate02}. Also notice that to distinguish a UMEB of Observation \ref{obs1} locally, one additionally requires a $d\otimes d$ maximally entangled state as resource \cite{Ghosh01, Horodecki03}. 

\subsection{Application: Methods to generate two-element LOCC-indistinguishable ensembles}\label{sec2subsec1}
We provide here a potential application of the concept of the UEB. Consider any two orthogonal mixed states, $\rho_1=p_1\ket{\phi_1}\bra{\phi_1}+p_2\ket{\phi_2}\bra{\phi_2}$ and $\rho_2=q_1\ket{\phi_3}\bra{\phi_3}+q_2\ket{\phi_4}\bra{\phi_4}$, where $\mathcal{B} = \{\ket{\phi_i}\}_{i=1}^4$ is a two-qubit orthonormal basis, $p_1+p_2=1=q_1+q_2$, \(p_1, p_2, q_1, q_2 \geq 0\). For such mixed states, we can state the following theorem.

\begin{theorem}\label{theo1} If $\mathcal{B}$ contains a two-qubit unextendible entangled basis, then the ensemble $\{\rho_1, \rho_2\}$ cannot be perfectly distinguished by separable measurements (and therefore, by LOCC). Nevertheless, the above ensemble is conclusively locally distinguishable. \end{theorem}

\noindent
{\it Proof.}~We have already proved that in a two-qubit system, the only possible cardinality for an unextendible entangled basis  is three (see Proposition \ref{prop1}). Hence, any two-qubit orthonormal basis $\mathcal{B}$ which contains a two-qubit UEB, must contain a product state and three entangled states. Without loss of generality, we assume that the states $\{\ket{\phi_i}\}_{i=1}^3$ are entangled states and the state $\ket{\phi_4}$ is a product state. We mention here that any convex mixture of an entangled pure state and a product state, always produces an entangled state \cite{Horodecki03-1}. This implies that $\rho_2$ is entangled and the projector $\mathbb{P}_2$ onto the support of $\rho_2$ is also inseparable. 

If $\mathbb{P}_2$ is inseparable, then the projector $\mathbb{P}_1$ onto the support of $\rho_1$ is inseparable too. This follows from a contradiction: Suppose, $\mathbb{P}_1$ is separable and it can be written as a mixture of product states. Now, 
\begin{equation}\label{eq6}
\mathbb{P}_2 = \mathbb{I}-\mathbb{P}_1, 
\end{equation}
where $\mathbb{I}$ is the $4\times4$ identity matrix. If we take partial transpose of the operator on the left-hand-side,  then it leads to a non-positive operator. On the other hand, the right-hand-side is positive under partial transpose, as $\mathbb{I}$ is invariant under partial transpose and the product states in the assumed separable decomposition of $\mathbb{P}_1$ will be mapped to another set of product states under the partial transpose operation.

However, it is known that the projectors $\mathbb{P}_1$ and $\mathbb{P}_2$ must be separable in order to distinguish the states $\rho_1$ and $\rho_2$ by separable measurements \cite{Chitambar14-1}. But this condition is not satisfied in the present case. Clearly, the ensemble $\{\rho_1, \rho_2\}$ cannot be perfectly distinguished by separable measurements. This also proves that the ensemble $\{\rho_1, \rho_2\}$ cannot be perfectly distinguished by local operations and classical communication as the set of all LOCC measurements $\subset$ the set of all SEP.

For the second part of the above theorem, we first mention that given two orthogonal quantum states $\rho_1$ and $\rho_2$, they can be conclusively locally distinguishable if and only if there exists two product state $\ket{\alpha_1}$ and $\ket{\alpha_2}$ such that $\bra{\alpha_i}\rho_j\ket{\alpha_i}$ = $\delta_{ij}$ holds \cite{Chefles04}. There is a product state, viz. $\ket{\phi_4}$ in the support of $\rho_2$,  and this product state must be orthogonal to $\rho_1$. Again, any two-dimensional subspace of a two-qubit Hilbert space must contain at least one product state \cite{Sanpera98}. Therefore, there must be a product state in the support of $\rho_1$ and this product state must be orthogonal to $\rho_2$. Hence, the states $\rho_1$ and $\rho_2$ must be conclusively locally distinguishable. These completes the proof of the above theorem. \hfill$\square$
\vskip 0.1 in

Let us state here a few points that are relevant to the above theorem. 
\begin{itemize}
\item Theorem~\ref{theo1} provides us a method to systematically generate two-element LOCC-indistinguishable ensembles of orthogonal quantum states. Actually, the indistinguishability is for the strictly larger class of separable quantum operations. Therefore,  like an unextendible product basis \cite{Bennett99, Divincenzo03} gives us a method to generate an ensemble having ``quantum nonlocality without entanglement'' \cite{Bennett99-1}, an unextendible entangled basis gives us a method to generate a two-element LOCC-indistinguishable ensemble. 
\item Two-element ensembles of multiparty orthogonal {\it pure} states are always locally distinguishable \cite{Walgate00, Virmani01, Ji05}.
\item It is always possible to consider two multiparty non-orthogonal pure states that will of course not be LOCC-distinguishable deterministically but may be so conclusively. However, being non-orthogonal states, the bit encoded in them can never be perfectly decoded, with any (global or local) quantum measurement, while using sufficient amount of entanglement as an extra resource, the bit encoded in the \(\{\rho_1, \rho_2\}\) ensemble can be perfectly decoded. 
\end{itemize}

The above remarks lead us to consider the following information processing task that can be implemented by using any \(\{\rho_1, \rho_2\}\) ensemble corresponding to a UEB. We consider three spatially separated parties: Enola, Mycroft, and Sherlock~\cite{Springer06}. Enola is connected to Mycroft and Sherlock via quantum channels. But there is no quantum channel between Mycroft  and Sherlock. In fact, Mycroft and Sherlock are only allowed to perform LOCC. In this setting, Enola wants to share one cbit of information with Mycroft and Sherlock. For the encoding, Enola can prepare a quantum system in a state which is chosen from a set of two orthogonal states, and then the quantum system can be distributed between Mycroft and Sherlock. There are however the following conditions for sharing the information. 
\begin{itemize}
\item[(i)] Mycroft and Sherlock cannot decode the information by using LOCC, but can do so perfectly  if provided with enough shared pure entangled states as  resource.
\item[(ii)] Mycroft and Sherlock face hefty penalties if they make an error in recognizing the cbit.
\item[(iii)] Equally prohibitive penalties are levied from Mycroft and Sherlock if they try to gather the cbit and fail.  
\end{itemize}
Item (i) requires that the cbit be encoded in two orthogonal states that are LOCC-distinguishable. Item (ii) implies that Mycroft and Sherlock will not attempt a minimum-error LOCC distinguishing protocol. Item (iii) implies that they will not attempt a conclusive local distinguishing protocol to discern the cbit. Notice that the ensembles of Theorem~\ref{theo1} meet the requirements in all the three items required in this task.

It may be opined that one can use ensembles that are made up of two orthogonal two-party (mixed) states that are not locally distinguishable either deterministically (perfectly) or conclusively. They can certainly be used but since they possess a higher level of nonlocality (in the sense of local indistinguishability of orthogonal states) than the ones mentioned in Theorem~\ref{theo1}, it is reasonable to presume that they are more costly than the latter ones. We will consider the latter ones below in relation to the result obtained in Theorem~\ref{theo2}. 

We mention here that it is not necessary to consider both the states $\rho_1$ and $\rho_2$ of Theorem~\ref{theo1} as mixed states. One can also consider an ensemble of a pure state and a mixed state, and still the relevant features present in the ensemble of two mixed states can be captured. Consider, for example, the UEB of (\ref{eq2}), and let us  take any state of the UEB as $\rho_1$ (which now is therefore pure). Let $\rho_2$ be a convex combination of the other two states of that UEB. Such an ensemble also has the property that they cannot be perfectly distinguished by SEP but they can be conclusively locally distinguishable. It follows from the same proof technique as given in case of Theorem~\ref{theo1}. An important feature of this ensemble is that it covers the minimum dimension of a Hilbert space, which is three in the present case (sum of the dimensions of the supports corresponding to $\rho_1$ and $\rho_2$), and possesses the requisite features, viz. being deterministic locally indistinguishable but conclusive locally distinguishable.

{\bf Generalization: higher dimensions and higher cardinalities.}
It is possible to develop a technique to construct mixed states as given in Theorem \ref{theo1}, also in higher dimensions. This technique starts with a previous discussion. We go back to Eq.~(\ref{eq4}). This is a $2\otimes2$ MEB which plays the role of a UMEB in $2\otimes3$. In this $2\otimes3$ Hilbert space, the product states which are orthogonal to the states of Eq.~(\ref{eq4}), are $\ket{02^\prime}$ and $\ket{12^\prime}$. If we consider any state from Eq.~(\ref{eq4}) and take any convex combination with $\ket{02^\prime}$, then the resulting state must be inseparable. We label such a state as $\rho_1$. We then consider another state from Eq.~(\ref{eq4}) and take any convex combination with $\ket{12^\prime}$. In this case also, the resulting state must be inseparable. We label this state as $\rho_2$. We next consider another state, $\rho_3$, which is produced by taking any convex combination of the remaining states of Eq.~(\ref{eq4}). These three states can be  distinguished perfectly by SEP, only when there exist three separable operators $\{\Pi_i\}_{i=1}^3$, such that Tr$(\rho_i\Pi_j)$ = $\delta_{ij}$, $\forall i,j=1,2,3$. For completeness, the sum of the operators $\Pi_i$, must be the identity operator, acting on the $2\otimes3$ Hilbert space. In the present case, each operator $\Pi_i$ must be a rank-2 operator and it must be contained within the support of $\rho_i$. But this is impossible because there are no rank-2 separable operators in the supports of $\rho_1$ and $\rho_2$. Thus, the states $\rho_1$, $\rho_2$, and $\rho_3$ cannot be perfectly distinguished by SEP. Interestingly, in the support of $\rho_i$, it is possible to find a product state for $\forall i=1,2,3$. This implies that the states $\rho_1$, $\rho_2$, and $\rho_3$ are conclusively locally identifiable \cite{Chefles04}. Following this technique, one can produce higher dimensional sets of high cardinality, the states of which hold similar properties like the ensemble described in Theorem~\ref{theo1}.

\begin{theorem}
\label{theo2}\textbf{\textit{Stronger nonlocality:}} There exist UEBs which can  lead to ensembles of two orthogonal quantum states that cannot be conclusively locally distinguished. \end{theorem} 

\noindent 
{\it Proof}: To construct an ensemble $\{\rho_1, \rho_2\}$ in such a way that the ensemble must not be conclusively identified by LOCC, one could use the UEB given in (\ref{eq1}). One can take $\frac{1}{\sqrt{2}}(\ket{01}-\ket{10})$ as $\rho_1$, and $\rho_2$ as any convex combination of the other two states of the UEB. This ensemble cannot be conclusively distinguished by LOCC because $\rho_1$ can never be conclusively locally identifiable. The argument for impossibility of conclusive identification of the state $\rho_1$, follows from the fact that it is not possible to find a product state with the  property that the product state is non-orthogonal to $\rho_1$ but orthogonal to $\rho_2$ \cite{Chefles04}. 
\hfill $\square$ 
\vskip 0.1 in

If we consider Theorems~\ref{theo1} and \ref{theo2} together, then, we obtain a classification among the two-qubit ensembles of cardinality two, viz.,  
\begin{itemize}
\item[(I)] the ensembles which cannot be perfectly distinguished by SEP but that ensemble is conclusively locally distinguishable, and
\item[(II)] the ensemble which cannot be conclusively distinguished by LOCC.
\end{itemize}
They can be used in information processing protocols as per necessity with respect to their nonlocality strength.

There are articles containing discussions related to locally indistinguishable sets of two orthogonal states in a two-qubit system. For example, one can go through  Refs.~\cite{Duan14, Chitambar14-1}. However, here we establish connections between such ensembles and UEBs. In fact, UEBs can be helpful in the systematic constructions of such ensembles. Moreover, the above classification provide us the opportunity to establish an order relation between the ensembles. For instance, the sets of Theorems~\ref{theo1} and \ref{theo2} are both nonlocal:  under the setting of {\it perfect} local discrimination of orthogonal quantum states, both sets are equally nonlocal. But conclusive local discrimination of the states provide us the privilege of identifying the more nonlocal sets. In particular, the sets considered in Theorem~\ref{theo2} are more nonlocal, compared to those in Theorem~\ref{theo1}.

Theorem~\ref{theo2} can also be seen from the following perspective. It is known that two orthogonal pure states can always be perfectly distinguished by LOCC \cite{Walgate00}. In fact, any two linearly independent pure multipartite states can always be conclusively distinguished by LOCC \cite{Ji05}. Theorem~\ref{theo2} provides ensembles of two orthogonal quantum states which does not cover the whole space and yet they are not conclusively locally distinguishable.

\section{Multipartite systems}\label{sec3}
Given a maximally entangled state in a bipartite system, it is always possible to transform the state to any state of the considered Hilbert space via LOCC \cite{Lo01, Vidal00, Nielsen99, Vidal99, Hardy99, Jonathan99}. But in a multipartite system (a system with more than two parties) there is no such state from which it is possible to get an arbitrary state via LOCC, even probabilistically. This is due to the existence of stochastic LOCC inequivalent classes \cite{Dur00, Verstraete02, Miyake03, Horodecki09-1, Guhne09, Das17}. In this sense, in a multipartite system, there is no state which plays the role of a maximally entangled state like in bipartite systems. With reference to this context, it is important to mention that there are some alternative concepts of maximally entangled states in multipartite systems, such as absolutely maximally entangled states \cite{Facchi08, Helwig12, Helwig13, Helwig13-1, Goyeneche15, Huber17, Huber18, Raissi18, Alsina19, Shen20}. For local transformation rules among multipartite entangled states, see e.g. Refs.~\cite{Dur00, Verstraete02, Miyake03, Bennett00, Leifer04, Kraus10, Kraus10-1, Turgut10, Mathonet10, Ribeiro11, Vicente13, Spee17, Neven20}. We stick to the notion of multipartite UEBs which contain only genuinely entangled states. 

In $\mathbb{C}^2\otimes\mathbb{C}^2\otimes\mathbb{C}^2$, we construct two different UEBs and analyze different properties of those UEBs. The first UEB consists of only W-type states  \cite{Dur00, Zeilinger92, Sen(De)03} but the second UEB consists of both Greenberger-Horne-Zeilinger (GHZ)-like states and W-type states. See Ref.~\cite{Dur00} for the structures of the states belonging to the GHZ-class and the W-class. The first UEB is constituted by the following states:
\begin{equation}\label{eq7}
\begin{array}{l}
\frac{1}{\sqrt{3}}(\ket{001}+\ket{010}+\ket{100}),\\[1.5 ex]
\frac{1}{\sqrt{3}}(\ket{001}+\omega\ket{010}+\omega^2\ket{100}),\\[1.5 ex]
\frac{1}{\sqrt{3}}(\ket{001}+\omega^2\ket{010}+\omega\ket{100}),\\[1.5 ex]
\frac{1}{\sqrt{3}}(\ket{000}+\ket{101}+\ket{110}),\\[1.5 ex]
\frac{1}{\sqrt{3}}(\ket{000}+\omega\ket{101}+\omega^2\ket{110}),\\[1.5 ex]
\frac{1}{\sqrt{3}}(\ket{000}+\omega^2\ket{101}+\omega\ket{110}).
\end{array}
\end{equation} 
Consider any one of the above states, and then let us trace out any of the qubits. The two-qubit reduced density matrix has only one product state in its range. So, the above states belong to the W-class \cite{Dur00}. We now present the following theorem.

\begin{theorem}\label{theo3}
The six states in (\ref{eq7}), belonging to the W-class, form an unextendible entangled basis, made of genuinely entangled states. Moreover, the unextendibility property of the basis remains conserved across every bipartition. 
\end{theorem} 

\noindent {\it Remark:} Here, by ``conserved'', we mean the carrying over of the property of unextendibility of the multiparty basis to the bipartition cases.  

\begin{proof}
Notice that there is a two-dimensional space that is orthogonal to the six states. It is spanned by the two states, $\ket{011}$ and $\ket{111}$, and they are separable across every bipartition. Taking any linear combination of these two fully separable states, it is not possible to generate any entangled states (neither biseparable states nor genuinely entangled states). So, the above six states not only form a three-qubit UEB, but also the unextendibility property of the UEB remains ``conserved'' across every bipartition. Obviously, if  linear combinations of the orthogonal fully separable pure states in the complementary subspace are able to produce biseparable states, then the given states must not form a UEB in at least one bipartition. 
\end{proof}

A general algorithm to produce UEBs in any multipartite system, whose unextendibility property should remain conserved across every bipartition, includes two steps: (i) finding a set of pure orthogonal genuinely entangled states which span a proper subspace of the considered Hilbert space, (ii) in the complementary subspace there should be only fully separable states. We believe that the result in Theorem \ref{theo3} is interesting, especially because there is no known example of an unextendible product basis (UPB) which is unextendible across every bipartition. We note here that a UPB is unextendible across every bipartition if it is not possible to get any product state in the complementary subspace considering any bipartition. On the other hand, a UEB is unextendible across every bipartition if it is not possible to get any entangled state in the complementary subspace considering any bipartition. 

The result of Theorem \ref{theo3} could also be seen in  light of the fact that product states including biseparable ones of a multipartite system form a set of measure zero, and almost all states are genuinely multisite entangled. In spite of this abundance of entangled states, there does exist a multiparty UEB whose unextendibility is conserved across every bipartition. On the other hand, despite the meagre presence of product states, a multiparty UPB with the same property has not as yet been found. 

In this context, we mention that in Ref.~\cite{Agrawal19}, a type of incomplete basis is constructed, termed as unextendible biseparable basis which cannot be completed by adding product states across every bipartition. Here, the notion is completely opposite. Here, the basis consists of genuinely entangled states such that the orthogonal complement contains only triseparable states. We now present the following proposition.

\begin{proposition}\label{prop3}
In the multiqubit configuration, it is not possible to construct an unextendible product basis which is unextendible across every bipartition while it is possible to construct an unextendible entangled basis whose unextendibility property remains conserved across every bipartition. 
\end{proposition}

\begin{proof}
The first part of the above proposition is due to the fact that in $\mathbb{C}^2\otimes\mathbb{C}^d$, there is no UPB \cite{Divincenzo03}. So, there is no multi-qubit UPB which is unextendible across every bipartition. The second part of the above proposition is due to Theorem \ref{theo3} and the Appendix.
\end{proof}

We now move to consider the cardinality (i.e., the size) of a multiparty UEB that remains a UEB in all partitions, and present the following proposition.

\begin{proposition}\label{prop4}
A three-qubit unextendible entangled basis which is unextendible across every bipartition can have the cardinality of six and seven.
\end{proposition}

\begin{proof}
In case of three qubits, any set of five or a lower number of pure mutually orthogonal genuinely entangled states cannot show unextendibility of the required kind, i.e., unextendibility that is retained in all partitions. This can be seen as follows. We consider any set of five pure mutually orthogonal genuinely multipartite entangled states, and assume that they form a UEB that remains a UEB in all partitions. So, in the remaining part of the multiparty Hilbert space (the orthogonal complement of the space spanned by the states of the UEB), one can always find at least three mutually orthogonal fully separable pure states. We now consider a particular bipartition, and take any linear combination of those three fully separable states, such that the coefficients are nonzero. By our assumption, the newly generated state must be separable in that bipartition. Such a state can be written as $\ket{a^\prime}(a_1\ket{a_1}+a_2\ket{a_2}+a_3\ket{a_3})$, where $|a_1|^2+|a_2|^2+|a_3|^3=1$ and \(a_1\), \(a_2\), \(a_3\) are nonzero. Again, the states $\ket{a_i}$ $\forall i = 1,2,3$, are pairwise orthogonal. According to our assumption, the two-qubit state \(a_1\ket{a_1}+a_2\ket{a_2}+a_3\ket{a_3}\) must be separable for all \(a_1\), \(a_2\), \(a_3\), and so can be expressed in the form  \(\ket{a^{\prime\prime}}(a_1^\prime\ket{a_1^\prime}+a_2^\prime\ket{a_2^\prime}+a_3^\prime\ket{a_3^\prime})\), and since \(|a_1\rangle\), \(|a_2\rangle\), \(|a_3\rangle\) are mutually orthogonal and \(a_1\), \(a_2\), \(a_3\) are nonzero, we must have that \(|a_1^\prime\rangle\), \(|a_2^\prime\rangle\), \(|a_3^\prime\rangle\) are mutually orthogonal and \(a_1^\prime\), \(a_2^\prime\), \(a_3^\prime\) are nonzero. This is a contradiction as the states \(|a_1^\prime\rangle\), \(|a_2^\prime\rangle\), \(|a_3^\prime\rangle\) belong to a qubit space. So, the two-qubit state $a_1\ket{a_1}+a_2\ket{a_2}+a_3\ket{a_3}$ might be an entangled state. Clearly, the original three-qubit state may not be separable across every bipartition. This shows that our initial assumption of the existence of a three-qubit UEB of cardinality five which moreover remains a UEB in all partitions was not true.

When the cardinality of the given set is less than five, then also it is possible to have a biseparable state in the complementary subspace. Hence, we arrive to the above proposition.
\end{proof}

Next, we present a second type of UEB for a three-qubit system. The states are given in the following list.
\begin{equation}\label{eq8}
\begin{array}{l}
\frac{1}{2}(\ket{000}+\ket{011}+\ket{101}+\ket{110}),\\[1.5 ex]
\frac{1}{2}(\ket{000}+\ket{011}-\ket{101}-\ket{110}),\\[1.5 ex]
\frac{1}{2}(\ket{000}-\ket{011}+\ket{101}-\ket{110}),\\[1.5 ex]
\frac{1}{2}(\ket{000}-\ket{011}-\ket{101}+\ket{110}),\\[1.5 ex]
\frac{1}{\sqrt{3}}(\ket{001}+\ket{010}+\ket{100}),\\[1.5 ex]
\frac{1}{\sqrt{3}}(\ket{001}+\omega\ket{010}+\omega^2\ket{100}),\\[1.5 ex]
\frac{1}{\sqrt{3}}(\ket{001}+\omega^2\ket{010}+\omega\ket{100}).\\[1.5 ex]
\end{array}
\end{equation}
Notice that in the above list, the first four states belong to the GHZ-class, while the remaining three  belong to the W-class. We now present the following theorem.

\begin{theorem}\label{theo4}
The states in (\ref{eq8})  form a three-qubit unextendible entangled basis of cardinality seven, and is unextendible across every bipartition. Also, implementing the measurement onto the complete basis corresponding to  the unextendible entangled basis is a nonlocal operation. Moreover, a local implementation of the nonlocal operation cannot be performed by a pure entangled resource state of the same dimensions as the basis states.
\end{theorem}
\begin{proof}
To complete the above basis, there is only one state left in the Hilbert space, which is $\ket{111}$, a fully separable state. Therefore, the above seven states form a UEB of maximum cardinality. It is also true that there is no biseparable state which is orthogonal to the above seven states. Therefore, the above UEB is also unextendible across every bipartition. 

An important property of the complete basis (which includes four states belonging to the GHZ-class, three states belonging to the W-class, and a fully separable state) is that distinguishing them corresponds to a nonlocal operation, and moreover, the operation  cannot even be implemented locally using any three-qubit entangled resource. The proof of this follows from the fact that the basis contains states from both the stochastic LOCC inequivalent classes of three-qubit pure states, and therefore, it is not possible to find a three-qubit pure resource state from which one can get all the basis states with some nonzero probability. But the non-availability of a single state in a certain multiparty Hilbert space that can be transformed to all the states (of that space) in a set with some nonzero probability is known to imply the non-existence of a resource state in that space for distinguishing the set of states  \cite{Bandyopadhyay16}. This implies that for the nonlocal operation to distinguish the states, one cannot use a three-qubit pure state as a resource.
\end{proof}

We now move to discuss an interesting local indistinguishability property for the UEB constituted by the states in (\ref{eq8}).

\begin{proposition}\label{prop5}
The unextendible entangled basis of cardinality seven formed by the states in (\ref{eq8}) is locally indistinguishable across every bipartition. Moreover, there exists a subset of five states which possesses such a property. 
\end{proposition}
\begin{proof}
Consider the first two GHZ-like states and the first W-type state in the list in (\ref{eq8}), and then view them in the first qubit vs. the rest bipartition. It is possible to project them in a two-qubit subspace with some nonzero probability. The two-qubit subspace is formed by the first qubit of the three-qubit system, and a two-dimensional subspace of the Hilbert space of the other two qubits.  This two-dimensional subspace is spanned by the vectors \(\ket{\phi^+}=(1/\sqrt{2})(\ket{00}+\ket{11})\) and \(\ket{\psi^+}=(1/\sqrt{2})(\ket{01}+\ket{10})\), of the second and third qubits. 

The projected states are pure orthogonal entangled states. Now, it is known that in a two-qubit system, three orthogonal pure entangled states cannot be perfectly distinguished by LOCC \cite{Walgate02}. This indicates that the UEB is locally indistinguishable in the first qubit vs. rest configuration. Following the same arguments and using the first and the third GHZ-like states along with the first W-type state, it is possible to prove that the UEB is locally indistinguishable in the second qubit vs. rest configuration. Similarly, using the first and the fourth GHZ-like states along with the first W-type state, it is possible to prove that the UEB is locally indistinguishable in the third qubit vs. rest configuration. So, we have derived that the UEB cannot be perfectly distinguished by LOCC across any of the three bipartitions. Moreover, the first five states possess the property that they cannot be perfectly distinguished by LOCC across any bipartition. 
\end{proof}

It is possible to construct a set of five orthogonal three-qubit GHZ states which is locally indistinguishable across every bipartition \cite{Zhang20}. The above proposition provides a set of five orthogonal three-qubit states which does not belong to the same inequivalent class but again the set is locally indistinguishable across every bipartition. 

A general algorithm to construct the UEBs of the ``second kind'', viz. UEBs of cardinality of unity less than the total dimension of the joint Hilbert space and which contain states from stochastic LOCC inequivalent classes, for any number of qubits is given in the Appendix. 

\section{Conclusion}\label{sec4}
We found interrelations between the phenomena of local indistinguishability and unextendibility of entangled bases of quantum states of bipartite and multipartite physical systems. Among the results obtained in the bipartite case, was one where we proved that the cardinality of a two-qubit unextendible entangled basis is always restricted to three, and that such bases are not distinguishable even by separable measurements, which is known to be a larger class of quantum operations than local quantum operations and classical communication. As an application of the results, we identified a method to generate ensembles of two orthogonal mixed states that are locally indistinguishable. Two orthogonal pure states were known to be always locally distinguishable. We point to quantum information sharing protocols where such ensembles can be potentially useful.

In the case of multipartite bases of quantum states, we have introduced the notion of unextendibility across every bipartition within unextendible entangled bases. We have  identified a class of unextendible entangled bases which lead to a class of nonlocal operations, local implementation of which require entangled resource states from a higher-dimensional Hilbert space. 

\begin{widetext}
\section*{Appendix}
\noindent \textbf{UEBs of the ``second kind'' for an arbitrary number of qubits:}
\vspace{0.5cm}

We first consider a four-qubit system. Now, consider the bit strings $0001$, $0010$, $0100$, $1000$. Using the corresponding fully separable pure states, one can consider the following ``four-qubit W states'':
\begin{equation}\label{eq9}
\begin{array}{c}
\frac{1}{2}(\ket{0001}+\ket{0010}+\ket{0100}+\ket{1000}),~~
\frac{1}{2}(\ket{0001}+\ket{0010}-\ket{0100}-\ket{1000}),\\[1.5 ex]
\frac{1}{2}(\ket{0001}-\ket{0010}+\ket{0100}-\ket{1000}),~~
\frac{1}{2}(\ket{0001}-\ket{0010}-\ket{0100}+\ket{1000}).
\end{array}
\end{equation}
Next, we consider the bit-wise orthogonal bit strings $1110$, $1101$, $1011$, $0111$. Using the corresponding fully separable pure states, one can consider the following states, 
again of the ``W-type'':
\begin{equation}\label{eq10}
\begin{array}{c}
\frac{1}{2}(\ket{1110}+\ket{1101}+\ket{1011}+\ket{0111}),~~
\frac{1}{2}(\ket{1110}+\ket{1101}-\ket{1011}-\ket{0111}),\\[1.5 ex]
\frac{1}{2}(\ket{1110}-\ket{1101}+\ket{1011}-\ket{0111}),~~
\frac{1}{2}(\ket{1110}-\ket{1101}-\ket{1011}+\ket{0111}).
\end{array}
\end{equation}
For a four-qubit system, there are a total of sixteen states in an orthogonal basis. The remaining eight orthogonal bit strings are $0000$, $1111$, $0011$, $1100$, $0101$, $1010$, $0110$, $1001$. Keeping aside the bit strings $0000$, $1111$, $0011$, $1100$, the other four can be used to construct four genuinely entangled states of the ``GHZ-type'' as follows: 
\begin{equation}\label{eq11}
\begin{array}{c}
\frac{1}{\sqrt{2}}(\ket{0101}\pm\ket{1010}),~~
\frac{1}{\sqrt{2}}(\ket{0110}\pm\ket{1001}).
\end{array}
\end{equation}
We next consider another three genuinely entangled states which are given as the following: 
\begin{equation}\label{eq12}
\begin{array}{c}
\frac{1}{\sqrt{2}}(\ket{0011}+\ket{1100}),~~\frac{1}{\sqrt{2}}(\frac{1}{\sqrt{2}}\ket{0011}-\frac{1}{\sqrt{2}}\ket{1100}\pm\ket{0000}).
\end{array}
\end{equation}
Notice that the first state of the above equation is a standard GHZ state - so were the ones in (\ref{eq11}) - and the other two states are also genuinely entangled states. So, now there is only one state left to complete the basis, and that is $\ket{1111}$, a fully separable state. Clearly, the states in (\ref{eq9})-(\ref{eq12}) form a four-qubit UEB of maximum cardinality, which is also unextendible across every bipartition. If we can now show that the complete basis contains states from two SLOCC (stochastic LOCC) inequivalent classes, it will follow that the basis will correspond to a nonlocal operation, to implement which by LOCC, it requires an entangled resource state from a higher-dimensional Hilbert space. Following the above process, it is easy to construct multi-qubit UEBs when the number of qubits $\geq5$. Modifying the steps it is also possible to produce UEBs of different cardinalities. 
\vspace{1cm}

\noindent \textbf{The \(N\)-qubit GHZ and W states are SLOCC inequivalent:}
\vspace{0.5cm}

We are now left with proving that the above basis contains states that belong to at least two SLOCC inequivalent classes. We will therefore show that the \(N\)-qubit GHZ state  \(|GHZ_N\rangle=\frac{1}{\sqrt{2}}(|0^{\otimes N}\rangle + |1^{\otimes N}\rangle)\)  and the \(N\)-qubit W state \(|W_N\rangle = \frac{1}{\sqrt{N}}(\sum |0^{\otimes (N-1)}1\rangle)\)  are SLOCC inequivalent. This result is well-known in the community, but we provide a proof of it for completeness. The proof directly follows from the arguments in Refs. \cite{Dur00, Sanpera98}. Let the parties sharing the \(N\)-qubit state be named as \(A_1\), \(A_2\), ..., \(A_N\). It was proven in Ref. \cite{Dur00} that any state that is SLOCC equivalent to the \(N\)-qubit GHZ state \(|GHZ_N\rangle\) can be expressed as \(|a_1\rangle_{A_1} |a_2\rangle_{A_2} \ldots |a_N\rangle_{A_N} + |b_1\rangle_{A_1} |b_2\rangle_{A_2} \ldots |b_N\rangle_{A_N}\), where \(|a_i\rangle\) and \(|b_i\rangle\) are vectors of the qubit Hilbert space associated with the system \(A_i\), with \(i = 1, 2, \ldots, N\). Suppose now that the state \(|W_N\rangle\) can be expressed as \(|a_1\rangle_{A_1} |a_2\rangle_{A_2} \ldots |a_N\rangle_{A_N} + |b_1\rangle_{A_1} |b_2\rangle_{A_2} \ldots |b_N\rangle_{A_N}\), for some vectors \(|a_i\rangle\) and \(|b_i\rangle\) of the qubit Hilbert space associated with the system \(A_i\), with \(i = 1, 2, \ldots, N\). Then, \(|a_1\rangle|a_2\rangle\) and \(|b_1\rangle |b_2\rangle\) will span \(R(\rho^W_{A_1A_2})\), the range of the local density matrix of \(|W_N\rangle\) after tracing out all parties except \(A_1\) and \(A_2\). Since the ranks of the local density matrices of \(\rho^W_{A_1A_2}\), which are just the single-qubit local densities of the \(N\)-qubit W state, are two each, \(R(\rho^W_{A_1A_2})\), being a two-dimensional subspace of \(\mathbb{C}^2 \otimes \mathbb{C}^2\), will contain exactly two product states \cite{Sanpera98}. Now, \(\rho^W_{A_1A_2} = \mbox{tr}_{A_3 \ldots A_N} |W_N\rangle \langle W_N| = \frac{1}{N}(2|\psi^+\rangle \langle \psi^+| + (N-2)|00\rangle \langle 00|)\), where \(|\psi^+\rangle = \frac{1}{\sqrt{2}}(|01\rangle + |10\rangle)\). It is easy to show that an arbitrary superposition of the states \(|\psi^+\rangle\) and \(|00\rangle\) has only one product state, viz. \(|00\rangle\), and therefore this is the only product state in \(R(\rho^W_{A_1A_2})\). This contradicts the assumption that \(|W_N\rangle\) can be written as \(|a_1\rangle_{A_1} |a_2\rangle_{A_2} \ldots |a_N\rangle_{A_N} + |b_1\rangle_{A_1} |b_2\rangle_{A_2} \ldots |b_N\rangle_{A_N}\), proving that the \(N\)-qubit W state is not SLOCC equivalent to the \(N\)-qubit GHZ state \(|GHZ_N\rangle\).
\end{widetext}
\bibliography{ref}
\end{document}